\documentclass{notices}

\usepackage{amsmath}

\usepackage{amsmath}   
\usepackage{amsfonts}

\usepackage{latexsym}
\usepackage{upgreek}
\usepackage{amssymb}

\usepackage{amsthm}
\usepackage{graphicx}
\usepackage{bm}
\usepackage{bbm}
\usepackage{dsfont}
\usepackage{xcolor}
\usepackage{comment}
\numberwithin{equation}{section}

\newcommand{\dkl}{D_{\mbox {\tiny{\rm KL}}}}

\newcommand{\dive}{\text{div}}

\newcommand{\R}{\mathbb{R}}
\newcommand{\veps}{\varepsilon}
\newcommand{\T}{\mathcal{T}}

\newcommand{\M}{\mathcal{M}}

\newtheorem{theorem}{Theorem}[section]

\newtheorem{proposition}[theorem]{Proposition}
\newcommand{\argmin}{\mbox{argmin}}
\theoremstyle{remark}
\newtheorem{remark}[theorem]{Remark}
\newtheorem{example}[theorem]{Example}

\theoremstyle{remark}
\newtheorem{remarkx}[theorem]{Remark}
 \renewenvironment{remark}
  {\pushQED{\qed}\remarkx}
  {\popQED\endremarkx}
  
\newtheorem{exx}[theorem]{Example}
 \renewenvironment{example}
  {\pushQED{\qed}\exx}
  {\popQED\endexx}

\newcommand{\bh}[1]{{\nc #1}}

\newcommand{\red}{\nc}
\newcommand{\nc}{\normalcolor}

\newcommand{\mY}{\mathcal{Y}}
\DeclareMathOperator*{\argmax}{arg\,max}

\title{From Optimization to Sampling Through Gradient Flows}

\author{
  N. Garc\'ia Trillos
  \affil{
   Assistant Professor, Department of Statistics, University of Wisconsin-Madison, Madison, WI 53706, USA, garciatrillo@wisc.edu
    }
  \and
  B. Hosseini
  \affil{
  Assistant Professor, Department of Applied Mathematics, University of Washington,   
Seattle, WA, USA, bamdadh@uw.edu
   }
   \and
   D. Sanz-Alonso
     \affil{
   Assistant Professor, Department of Statistics, University of Chicago, Chicago, IL 60637, USA, sanzalonso@uchicago.edu
   }
}

\begin{document}

\maketitle

Optimization and sampling algorithms play a central role in science and engineering as they enable finding optimal predictions, policies, and recommendations, as well as expected and equilibrium states of complex systems. The notion of `optimality' is formalized  by the choice of an objective function, while the notion of an `expected' state is specified by a probabilistic model for the distribution of states. Optimizing rugged objective functions and sampling multi-modal distributions is computationally challenging, especially in high-dimensional problems. For this reason, many optimization and sampling methods have been developed by researchers working on disparate fields such as Bayesian statistics, molecular dynamics,  genetics, quantum chemistry, machine learning, weather forecasting, econometrics, and medical imaging.

State-of-the-art algorithms for optimization and sampling often rely on \emph{ad-hoc} heuristics and empirical tuning, but some unifying principles have emerged that greatly facilitate the understanding of these methods and the  communication of algorithmic innovations across scientific communities. This article is concerned with one such principle: the use of gradient flows, and discretizations thereof, to  design and analyze optimization and sampling algorithms. The interplay between optimization, sampling, and gradient flows is an active research area and a thorough review of the extant literature is beyond the scope of this article\footnote{AMS Notices limits to $20$ the references per article;  we refer to  \cites{trillos2020bayesian,chewi2020exponential,garbuno2020interacting} for further pointers to the literature.}.  Our goal is to provide an accessible and lively introduction to some core ideas, emphasizing that gradient flows uncover the conceptual unity behind several existing algorithms and give a rich mathematical framework for their rigorous analysis. 
 
We present motivating applications in section~\ref{sec:motivation}. 
section~\ref{sec:GradDescent} is focused on the gradient descent approach to optimization and introduces fundamental ideas such as preconditioning, convergence analysis, and 
time discretization of gradient flows. Sampling is discussed in section~\ref{sec:Sampling}
in the context of Langevin dynamics viewed as a gradient flow of the Kullback-Leibler (KL) divergence 
with respect to (w.r.t) the Wasserstein geometry. Some modern applications of 
gradient flows for sampling are discussed in section~\ref{sec:otherGFs},
followed by concluding remarks in section~\ref{sec:conclusion}.

\section{Motivating Applications}\label{sec:motivation}
We outline two applications in Bayesian statistics and molecular dynamics that illustrate some important challenges in optimization and sampling. 
\subsection{Bayesian Statistics}
\label{sec:Bayesian}
  In Bayesian statistics \cite{gelman1995bayesian}, an initial belief about 
 an unknown parameter is updated as data becomes 
 available.  Let $\theta$ denote the unknown 
 parameter of interest belonging to the parameter space $\Theta$
  and let $\pi_0(\theta)$ denote the \emph{prior} distribution 
  reflecting our initial belief. Furthermore, let $y$ 
  be the observed data also belonging to an appropriate space $\mY$. 
  Then Bayes' rule identifies the distribution of $\theta$ conditioned on the data $y:$ 
  \begin{equation}\label{bayes-rule}
      \pi(\theta | y) \propto \pi(y | \theta) \pi_0(\theta),
  \end{equation}
  where $\propto$ indicates that the right-hand side should 
  be normalized  to define a probability distribution. 
  Here $\pi(y | \theta)$ is called the \emph{likelihood} function
  and $\pi(\theta | y)$ is called the \emph{posterior} distribution. 
  Bayesian inference on $\theta$ is based on 
  the posterior, which blends the information  in the  
  prior and the data.

  The choice of prior and likelihood
  is a modeling task which, perhaps surprisingly, is 
  often not the most challenging aspect of Bayesian inference.  
  The main challenge is to extract information from  
  the posterior since (i) it typically does not belong to a standard  family of distributions,
  unless in the restrictive case of conjugate models \cite{gelman1995bayesian}; (ii) the 
  parameter $\theta$ can be high dimensional; and (iii) the normalizing constant $\int_\Theta \pi(y|\theta) \pi_0(\theta)\, d\theta$ in \eqref{bayes-rule} (known as the {\it marginal likelihood}) is rarely available and it can be expensive to compute. 
  These practical hurdles inform the design of optimization and sampling algorithms to find posterior statistics.

  Of particular importance is  the posterior mode or 
  \emph{maximum a posteriori (MAP)} estimator
  $$\theta_{\mbox {\tiny{\rm MAP}}} 
   := \argmax_\theta \pi(\theta | y).$$
  Many optimization algorithms for MAP estimation start from an initial guess $\theta_0$ and produce iterates $\{\theta_n\}_{n=1}^N$ by discretizing a gradient flow with the property that $\theta_n \approx  \theta_{\mbox {\tiny{\rm MAP}}} $ for large $n$. Such  gradient flows in parameter space 
  will be discussed in section \ref{sec:GradDescent}.
  
    Computing MAP estimators is closely related to classic regularization 
  techniques such as penalized least squares and Tikhonov regularization \cite{sanz2018inverse}. 
  In order to fully leverage the Bayesian framework, it is often desirable to consider other posterior statistics such as mean, variance, credible intervals, and task-specific functional moments, 
  which can be written in the form
  $$ \mathbb{E}^{\pi(\cdot | y)} [ \phi(\theta)] := \int \phi(\theta) \pi(\theta | y) \, d\theta,$$
  where $\phi$ is a suitable test function. Since $\theta$ is often high dimensional, the standard approach to compute these expectations  is to use Monte Carlo \cite{liu2001monte}: 
   obtain $N$ samples $\{\theta_n\}_{n=1}^N$
  from the posterior $\pi(\theta | y)$, and then approximate
  $$ \mathbb{E}^{\pi(\cdot | y)} [ \phi(\theta)] \approx  \frac1N \sum_n \phi(\theta_n).$$
  While Monte Carlo integration is in principle scalable to high dimensions, the task of generating posterior samples is still highly non-trivial. To that end, one may consider sampling $\theta_n \sim \rho_n,$ where the sequence $\{\rho_n\}_{n=1}^N$ arises from discretizing a gradient flow with the property that $\rho_n \approx \pi(\cdot| y) $ for large $n.$
  Such gradient flows in the space of probability distributions will be discussed in sections \ref{sec:Sampling} and \ref{sec:otherGFs}.
  
 To relate the discussion above to subsequent developments, we note that dropping the data $y$ from the notation, the posterior density can be written as 
 \begin{equation}\label{V-definition}
     \pi(\theta) = \frac{1}{Z} \exp \bigl( - V(\theta) \bigr),
 \end{equation}
 where $Z = \int_\Theta \pi(y|\theta) \pi_0(\theta)\, d\theta$ is the marginal likelihood and $V: \Theta \times \mY  \to \R$ is the negative logarithm of the posterior density. 

 \subsection{Molecular Dynamics}
  Another important source of challenging optimization and sampling problems is statistical mechanics, and in particular the simulation of molecular dynamics (see chapter 9 in \cite{liu2001monte}). According to Boltzmann and Gibbs, the positions $q$ and momenta $p$ of the atoms in a molecular system of constant size, occupying a constant volume, and in contact with a heat bath (at constant temperature), are distributed according to
\begin{equation}\label{eq:boltzmann}
\pi(q,p) = \frac{1}{Z} \exp\Bigl(-\beta \bigl(U(q)+K(p)\bigr)\Bigr),
\end{equation}
where $Z$ is a normalizing constant known as the partition function, $\beta$ represents the inverse temperature,  $U$ is a potential energy describing the interaction of the particles in the system, and $K$ represents the kinetic energy of the system. Letting $\theta := (q,p),$ we can write the Boltzmann-Gibbs distribution \eqref{eq:boltzmann} in the form \eqref{V-definition}, with
\begin{equation}\label{eq:Vstatmech}
    V(\theta) = -\beta \bigl(U(q)+K(p)\bigr).
\end{equation}

As in Bayesian statistics, it is important to determine the most likely configuration 
of particles (i.e. the mode of $\pi$), along with expectations of certain test functions w.r.t. the Boltzmann distribution. These two tasks motivate the need for optimization and sampling algorithms that acknowledge that the potential $U$ is often a rough function with many local minima, that the dimension of $q$ and $p$ is large, and that finding the normalizing constant $Z$ is challenging.


\section{Optimization}
\label{sec:GradDescent}
In this section we discuss gradient flows for solution of the
unconstrained minimization problem 
\begin{equation}
 \textrm{minimize} \, V(\theta) \quad \textrm{s.t.} \,\, \theta \in \Theta, 
  \label{eq:MinV}
\end{equation}
where $V(\theta)$ is a given objective function. Henceforth we take $\Theta := \R^d$ unless otherwise noted. As guiding examples, consider computing the mode of a posterior or Boltzmann distribution by minimizing $V$ given by \eqref{V-definition} or \eqref{eq:Vstatmech}. The methods described in this section are applicable beyond the specific problem of finding the modes, however, this interpretation will be of particular interest in relating the material in this section to our discussion of sampling in section \ref{sec:Sampling}.


\subsection{Gradient Systems}
\label{sec:GradSystems}
One of the most standard approaches to solve \eqref{eq:MinV} is \textit{gradient descent}, an optimization scheme that is based on the discretization of the gradient system
\begin{equation}
\label{eq:GradDescent}  
\dot{\theta}_t = - \nabla V(\theta_t), \quad t >0,
\end{equation}
with user-defined initial value $\theta_0$; throughout this article $\nabla V(\theta_t)$ will denote the gradient of the function $V$ at the point $\theta_t$, which will be tacitly assumed to exist wherever needed. 

While equation \eqref{eq:GradDescent} is perhaps the most popular formulation of the continuous-time gradient descent dynamics, the equivalent integral form below reveals more transparently some of its properties:
\begin{equation}
V(\theta_t) =  V(\theta_s) - \frac{1}{2}\int_{s}^t |\nabla V (\theta_r)|^2 dr  - \frac{1}{2} \int_{s}^t |\dot \theta_r|^2 dr, 
\label{eq:EDE}
\end{equation}
for all $t\geq s >0$. Indeed, notice that from \eqref{eq:EDE} it is apparent that  $V(\theta_t) \leq V(\theta_s)$ for $s \leq t$, i.e.,  the value of the function $V$ decreases in time, and in all but a few trivial situations the decrease is strict.  Another advantage of the reformulation \eqref{eq:EDE} (or its inequality form \eqref{eq:EDI} below) is that it can be adapted to more general settings with less mathematical structure than the one needed to make sense of \eqref{eq:GradDescent}. In particular, \eqref{eq:EDE} can be used to motivate a
notion of gradient flow in arbitrary metric spaces; see \cite{ambrosio2008gradient} for an in-depth discussion of this topic.

\begin{proposition}
\label{prop:EDE}
Suppose $V$ is a $C^1$ function. Then \eqref{eq:GradDescent} and \eqref{eq:EDE} are equivalent  
and they both imply
\begin{equation} 
|\nabla V ( \theta_t)|^2 = |\dot{\theta}_t|^2. 
\label{eq:DissipationAux}
\end{equation}
\end{proposition}
\begin{proof}
By Cauchy-Schwartz and Young's inequalities, for any $t>0$ it holds that
\begin{align*}
    \red - \nc \langle \nabla V(\theta_t), \dot{\theta}_t \rangle &\le | \nabla V(\theta_t) | |\dot{\theta}_t | \\
    & \le  \frac12 |\nabla V(\theta_t) |^2 + \frac12 |\dot{\theta}_t |^2,
\end{align*}
and both inequalities are equalities iff $-\nabla V(\theta_t) = \dot{\theta}_t.$ 
Therefore, 
\begin{align*}
    V(\theta_t) &= V(\theta_s) + \int_s^t \langle \nabla V(\theta_r), \dot{\theta}_r \rangle dr \\
                & \ge  V(\theta_s) - \frac12 \int_s^t |\nabla V(\theta_r)|^2 dr - \frac12 \int_s^t |\dot{\theta_r}|^2 dr,
\end{align*}
and equality holds iff $-\nabla V(\theta_t) = \dot{\theta}_t$ for all $t>0.$  The identity $|\nabla V(\theta_t)|^2 = |\dot{\theta}_t|^2$ follows directly from \eqref{eq:GradDescent}.
\end{proof}

Notice that the relationship $\dot \theta_t = - \nabla V(\theta_t)$ is only required in proving the \textit{energy dissipation inequality}
\begin{equation}
V(\theta_t) \leq  V(\theta_s) - \frac{1}{2}\int_{s}^t |\nabla V (\theta_r)|^2 dr  - \frac{1}{2} \int_{s}^t |\dot \theta_r|^2 dr, 
\label{eq:EDI}
\end{equation}
since the reverse inequality is a consequence of Cauchy-Schwartz.
Notice further that \eqref{eq:GradDescent} implies \eqref{eq:DissipationAux}, but in general the converse statement is not true. For example, the flow $\dot{\theta_t} = \nabla V (\theta_t)$ satisfies \eqref{eq:DissipationAux}, but in general does not satisfy \eqref{eq:GradDescent}. Likewise, \eqref{eq:GradDescent} implies $\frac{d}{dt}V(\theta_t) =- |\nabla V(\theta_t)|^2$ (which follows directly from the chain rule), but not conversely. Indeed, in $\R^2$ we may take $A$ to be any orthogonal matrix and consider  $\dot{\theta}_t = -A^2 \nabla V(\theta_t)$ so that $\frac{d}{dt}V(\theta_t) =- |\nabla V(\theta_t)|^2$ but  \eqref{eq:GradDescent} is not, in general, satisfied.
This digression illustrates that equation \eqref{eq:EDE} captures in one single identity of \textit{scalar} quantities the \textit{vectorial} identity \eqref{eq:GradDescent}, even if it is not as intuitive as other scalar relations. 

\subsection{A Note on Convergence}\label{ssec:convODE}
Despite the fact that gradient descent satisfies the energy dissipation property, it is in general not true that as time goes to infinity the dynamics \eqref{eq:GradDescent} converge to a global minimizer of \eqref{eq:MinV}. This could happen for different reasons. First, the problem \eqref{eq:MinV} may not have a minimizer (e.g., take $V(\theta) = e^{-\theta}$ for $\theta \in \R$). Second, $\nabla V$ may have critical points associated with saddle points or local optima of $V$ as we illustrate in the next example.

\begin{example}
Consider the double well potential 
\begin{equation}\label{eq:doublewell}
    V(\theta) = \frac{3}{8}\theta^4 - \frac{3}{4}\theta^2, \quad \theta \in \R,
\end{equation}
so that $\nabla V(\theta) = \frac{3}{2} \theta(\theta^2 -1).$ Notice that $\theta = 0$ is an unstable equilibrium of \eqref{eq:GradDescent}, which corresponds to a local maximum of $V.$ For each local minima $\theta = \pm 1$ of $V$ there is an associated subregion in the space of parameters (known as a basin of attraction) such that any initial condition $\theta_0$ chosen in this subregion leads the gradient dynamics toward its corresponding local minimizer, see Figure \ref{fig:deterministic}.
\end{example}

\begin{figure}
    \centering
    \includegraphics[width = 0.8 \linewidth]{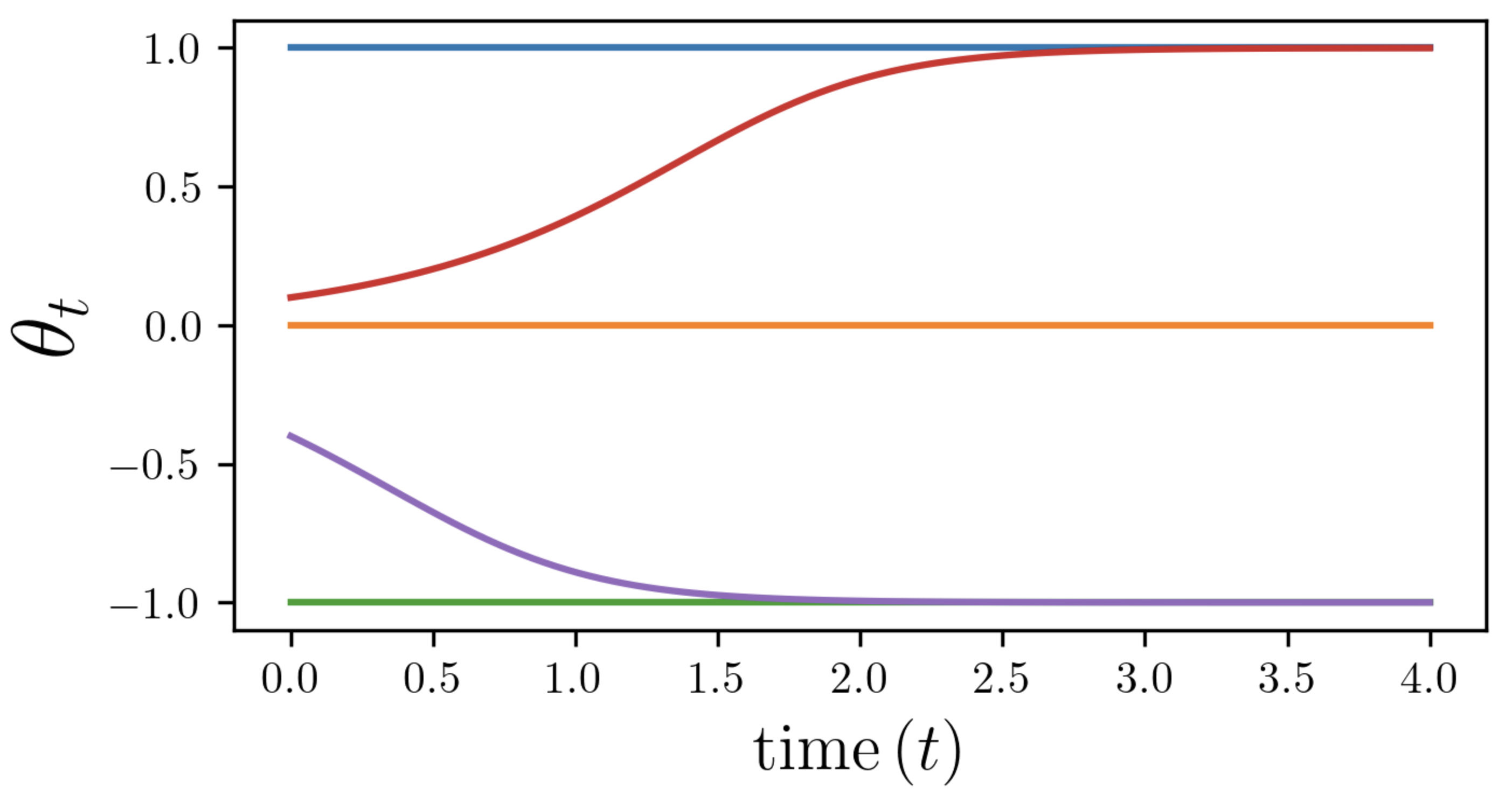}
   \vspace{-2ex} 
   \caption{Five trajectories of the one-dimensional gradient system \eqref{eq:GradDescent} with double well potential \eqref{eq:doublewell}. The objective $V$ has critical points at $\theta = 0$ and $\theta = \pm 1$. The intervals $(0,\infty)$ and $(-\infty,0)$ are basins of attraction for $\theta = 1$ and $\theta = -1,$ respectively.}
    \label{fig:deterministic}
\end{figure}

Suitable assumptions on $V$ prevent the existence of local minimizers that are not global and also imply rates of convergence. One such assumption 
is the Polyak-Lojasiewic (PL) condition \cite{karimi2016linear}:
\begin{equation}
\alpha (  V(\theta) - V^*) \leq \frac{1}{2}|\nabla V (\theta)|^2 , \quad \forall \theta \in \Theta,
\label{eq:PLCondition}
\end{equation}
where $V^\ast = \inf_{\theta \in \Theta} V(\theta)$ and  $\alpha > 0$ is a constant. Note that the PL 
condition readily implies that any stationary point $\theta^\ast$ of $V$ is a global minimizer, since 
\begin{equation*}
    \alpha (V (\theta^\ast) - V^\ast) \le \frac12 | \nabla V(\theta^\ast) |^2 =0.
\end{equation*}

Under the PL condition we can easily obtain a convergence rate for continuous-time gradient descent.
\begin{proposition}\label{prop:ExpDecay}
Suppose  $V$ satisfies  \eqref{eq:PLCondition}. Then, 
\begin{equation}
\label{eq:ExpDecay}
 V(\theta_t) - V^* \leq ( V(\theta_0) - V^* ) \exp(- 2  \alpha t ), \hspace{1ex} \forall t \ge 0.  
\end{equation}
\end{proposition}
\begin{proof}
Using condition \eqref{eq:PLCondition} in equation \eqref{eq:EDI} and recalling 
 \eqref{eq:DissipationAux} we conclude that
\[ V(\theta_t) - V^* \leq  (V(\theta_0)- V^*) -   2 \alpha \int_{0}^t  ( V(\theta_r) - V^*  )dr.  \]
The result follows by Gronwall's inequality. 
\end{proof}
One can verify the PL condition under various assumptions on the function $V$ \cite{karimi2016linear}. 
Here we present, as an important example, the case of $\alpha$-strong convexity.  For $\alpha>0,$ one says that $V$ 
is $\alpha$-\textit{strongly convex} if for any $\theta, \theta' \in \Theta$ it holds that
\[ 
\begin{aligned}
V  \bigl( t \theta & +(1-t)  \theta' \bigr) \leq  \\ 
& t V(\theta) + (1-t) V(\theta') - \frac{\alpha}{2} t(1-t) |\theta- \theta'|^2,
\end{aligned}
\]
for all $t \in [0,1]$. This condition can be shown to be equivalent to
\begin{equation*}
    V(\theta' ) \ge V(\theta) + \langle \nabla  V(\theta),\theta' - \theta \rangle + \frac{\alpha}{2} | \theta' - \theta|^2, \hspace{1ex} \forall \theta, \theta' \in \Theta,
\end{equation*}
from which we can see, after minimizing both sides w.r.t. $\theta'$, that
\begin{equation}
    V^\ast \ge V(\theta) - \frac{1}{2\alpha} | \nabla V(\theta)|^2,
    \label{eq:StrongConvexity}
\end{equation}
which is equivalent to \eqref{eq:PLCondition}. From this we conclude that $\alpha$-strong convexity implies the PL condition with the same constant $\alpha$.

Note that
strong convexity is a stronger condition than the PL condition. For example, the function $V(\theta)=\frac{1}{2}\theta_1^2$ (where $\theta = (\theta_1, \theta_2)$) satisfies the PL condition with $\alpha =1$, but it is not strongly convex.



\subsection{Choice of the Metric}
\label{sec:Precondition}

Let us consider a function of the form  $V(\theta) = \alpha_1 \theta_1^2 + \alpha_2 \theta_2^2,$ where $\theta = (\theta_1,\theta_2)\in \R^2.$ Suppose that $0 <\alpha_1 \ll \alpha_2 $ and that $\alpha_1$ is very close to zero, as in Figure \ref{fig:ellipse}. 
We can now apply Proposition \ref{prop:ExpDecay} with $\alpha=\alpha_1$, but since we assumed
$\alpha_1$ is small we see that the right-hand side of \eqref{eq:ExpDecay} decreases very slowly. This suggests that gradient descent may take a long time to reach $V$'s global minimum when initialized arbitrarily. 

\begin{figure}
    \centering
    \includegraphics[width = 0.9 \linewidth]{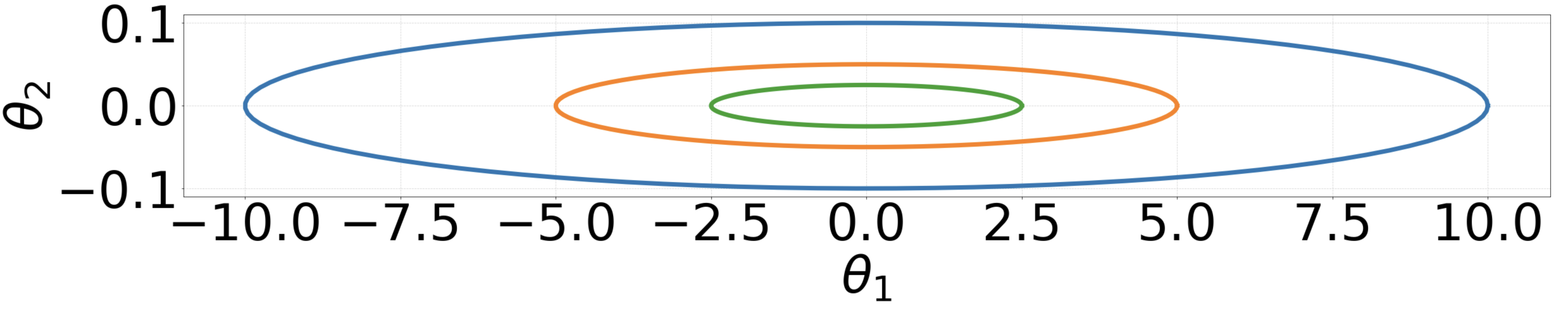}
   \vspace{-2ex} \caption{Level curves of a two-dimensional potential of the form $V(\theta) = \alpha_1 \theta_1^2 + \alpha_2 \theta_2^2$ with $0 <\alpha_1 \ll \alpha_2 $.  The anisotropy of this potential causes the gradient system \eqref{eq:GradDescent} to converge slowly.}
    \label{fig:ellipse}
\end{figure}
{
 The poor behavior of gradient descent described above 
arises whenever there are regions of points away from the minimizer at which the gradient of $V$ is very small.} 
One approach to remedy this issue is to introduce a more general version of gradient descent that accelerates the dynamics in those regions where the gradient of $V$ is small. This is the {goal of} \textit{preconditioning}. 
Let $H: \Theta \rightarrow \mathcal{S}^d_{++}$ be a continuous field of positive definite matrices, i.e., a 
function that assigns to every point $\theta \in \Theta$ a $d\times d$ positive definite matrix $H(\theta)$. The \textit{preconditioned gradient descent} dynamics induced by $H$ is defined as:
\begin{equation}
\label{eq:PreCondGradDescent}  
\dot{\theta}_t = - H(\theta_t)^{-1} \nabla V(\theta_t), \quad t >0.
\end{equation}
Observe that \eqref{eq:PreCondGradDescent} coincides 
with the original gradient descent dynamics \eqref{eq:GradDescent} when $H$ is constant and equal to the $d\times d$ identity matrix. On the other hand,  when $V$ is convex and twice differentiable, choosing 
$H(\theta) = \nabla^2 V(\theta)$, the Hessian of $V$, results in a continuous-time analog of {\it Newton's algorithm}.
In the example in Figure \ref{fig:ellipse}, we can directly compute  $\nabla^2 V=\left(\begin{matrix} 2 \alpha_1 & 0 \\ 0 & 2 \alpha_2 \end{matrix}\right)$, i.e., the Hessian is a fixed matrix since the potential $V$ is quadratic.
Substituting this choice of $H$ in \eqref{eq:PreCondGradDescent} for that example gives the dynamics $ \dot{\theta}_t = -\theta_t$, a scheme that achieves a much faster convergence rate.

 The reader may wonder if we could have in fact chosen $H = \frac{1}{r} D^2 V$ for large constant $r$ in order to induce a system that converges to equilibrium at a faster rate.  However, as implied by our discussion in section \ref{sec:discODE} and, specifically, Remark \ref{rem:FasterRatesVsDiscretization}, there is no benefit in doing so as the cost of discretizing becomes correspondingly higher; rescaling the Hessian may be simply interpreted as a change of units. In general, there is a natural tension between accelerating continuous-time dynamics by changing the metric of the space, and producing accurate time discretizations for the resulting flows; see \cite{trillos2020bayesian} for a related discussion in the context of sampling algorithms. In a similar vein, we notice that the superior convergence rate and affine invariance of Newton's algorithm comes at the price of utilizing the Hessian matrix $\nabla^2 V$, 
which in many applications can be prohibitively costly to compute or store. To this end, constructing matrix fields $H(\theta)$ that are good proxies for the Hessian and that can be computed efficiently is the goal of preconditioning. Perhaps the most 
well-known family of such algorithms is the family of {\it Quasi-Newton} algorithms \cite{nocedal} and 
in particular the \emph{Broyden–Fletcher–Goldfarb–Shanno} (BFGS) algorithm, which approximates the Hessian using gradients calculated at previous iterates. 


\subsubsection{Geometric Interpretation}
As a step toward introducing the material in section \ref{sec:Sampling}, here we give a geometric interpretation of equation \eqref{eq:PreCondGradDescent}. Specifically, we will show that \eqref{eq:PreCondGradDescent} can still be understood as a gradient descent equation but w.r.t. a different metric on the parameter space $\Theta$. For this purpose it is convenient to recall that a \textit{Riemannian manifold} $(\M, g)$ is a manifold $\M$ endowed with a family of inner products {$g=\{ g_{\theta} \}_{\theta \in \M}$} (often referred to as the \textit{metric}), one for each point on the manifold, and which can be used to measure angles between vectors at every tangent plane of $\M$. We will use $\T_\theta \M$ to denote the tangent plane at a given $\theta \in \M$. The standard example of a Riemannian manifold is $\M= \R^d$ with $g_\theta$ the Euclidean inner product at every  point. More general examples of Riemannian manifolds with $\M = \R^d$ can be generated from a field of positive definite matrices $H$. Consider the family of inner products:
\[ g_\theta (u,v) := \bigl\langle H(\theta) u, v \bigr\rangle, \quad u, v \in \R^d, \]
where $\langle \cdot, \cdot \rangle$ denotes the standard Euclidean inner product. 
Notice that $g_\theta$ is indeed an inner product since  $H(\theta)$ is  positive definite.
In what follows we often suppress the dependence of $g$ on $\theta$ for brevity.

We now proceed to define the notion of the gradient of a function $F$ defined over an arbitrary Riemannian manifold. Let $(\M, g)$ be a Riemannian manifold and let $F: \M \rightarrow \R$ be a smooth enough function. The gradient of $F$ at the point $\theta \in \M$ relative to the metric $g$, denoted $\nabla_g F (\theta),$ is defined as the vector in $\T_\theta \M$ for which the following identity holds:
\begin{equation}
   \frac{d}{dt} F \bigl(\gamma(t)\bigr) \big|_{t=0} = g\bigl(\nabla_g F(\theta), \dot{\gamma}(0)\bigr) 
   \label{eq:DefGradient}
\end{equation}
for any differentiable curve $\gamma:(-\veps, \veps) \rightarrow \M$ with {$\gamma(0)=\theta$}; by $\dot{\gamma}(0)$ we mean the velocity of the curve $\gamma$ at time $0$, which is an element in $\T_\theta \M$. In the example of $\M=\R^d$ with inner products induced by a field $H$ (which we denote with $g_H$), {we have that
\begin{align*}
\frac{d}{dt} V \bigl(\gamma(t)\bigr) \big|_{t=0} & = \bigl\langle \nabla V (\theta) , \dot{\gamma}(0) \bigr\rangle  
\\& = \bigl\langle H(\theta) H(\theta)^{-1} \nabla V (\theta) , \dot{\gamma}(0) \bigr\rangle
\\& = g \bigl( H(\theta)^{-1} \nabla V (\theta) , \dot{\gamma}(0)   \bigr),
\end{align*}
for any curve $\gamma: (-\varepsilon, \varepsilon) \rightarrow \M$ with $\gamma(0)=\theta$, from where it follows that $\nabla_{g_H} V (\theta) = H(\theta)^{-1} \nabla V (\theta)$, where we recall $\nabla$ denotes the usual gradient in $\R^d$. 

In this light, \eqref{eq:PreCondGradDescent} can be interpreted as a gradient descent algorithm, only that the gradient is taken 
w.r.t. a metric that is different from the standard Euclidean one.
In section \ref{sec:Sampling}, where we discuss sampling, we will return to some of the insights that we have developed in this section. In particular, in order to define gradient descent dynamics of a functional over a manifold we need to specify
two ingredients: 1) an energy $V$ to optimize, and 2) a metric $g_\theta$ under which we  define the  gradient. For the last item, it will be convenient to have a clear understanding of how to represent smooth curves in the manifold of interest and characterize their velocities appropriately. Indeed, equation \eqref{eq:DefGradient} explicitly relates the metric $g_\theta$ of the manifold, the target energy $V$, the rate of change of the energy  along arbitrary smooth curves $\gamma$, and the gradient $\nabla_g V$ of the energy relative to the chosen metric.}

\subsubsection{Geodesic Convexity}
There are analogous conditions to the PL and strong convexity assumptions discussed in section \ref{sec:GradSystems} that guarantee the convergence of the flow \eqref{eq:PreCondGradDescent} toward global minima of $V$. First, write \eqref{eq:PreCondGradDescent} as an energy dissipation equality of the form
\begin{equation}
V(\theta_t) =  V(\theta_s) - \frac{1}{2}\int_{s}^t |\nabla_g V (\theta_r)|_{\theta_r}^2 dr  - \frac{1}{2} \int_{s}^t |\dot \theta_r|_{\theta_r}^2 dr, 
\label{eq:EDEPrecond}
\end{equation}
where we have used $|\cdot|_\theta^2 $ to denote $g_\theta(\cdot, \cdot)$. The equivalence between \eqref{eq:PreCondGradDescent} and \eqref{eq:EDEPrecond} follows from an identical argument as in Proposition \ref{prop:EDE} applied to an arbitrary inner product. The analogous PL condition in the preconditioned setting takes the form: 
\begin{equation}
    \alpha (V (\theta) - V^\ast) \le \frac12 | \nabla_g V(\theta) |_\theta^2,
    \label{eq:PLGeneral}
\end{equation}
which generalizes the PL condition in the Euclidean setting to general inner products and gradients.

To introduce an appropriate notion of convexity that allows us to generalize the results of
section \ref{sec:GradSystems} we need to introduce a few more ideas from Riemannian geometry. Given a Riemannian manifold $(\M,g)$, we define the geodesic distance $d_g$ induced by the metric $g$ as:
{
\begin{equation}
   d_g^2(\theta, \theta') = \inf_{t \in [0,1] \mapsto (\gamma_t, \dot{\gamma}_t) }  \int_{0}^1   |\dot{\gamma}_t|_{\gamma_t}^2 dt. 
   \label{eq:GeoDistance}
\end{equation}
We will say that $\gamma: [0,1] \rightarrow \M$ is a constant speed geodesic between $\theta$ and $\theta'$ if $\gamma$ is a minimizer of the right-hand side of the above expression. Equivalently, 
a constant speed geodesic between $\theta$ and $\theta'$ is any curve with $\gamma(0)=\theta $ 
and $\gamma(1)=\theta'$ such that $d_g \bigl(\gamma(s), \gamma(t)\bigr) = |t-s| d_g(\theta, \theta')$ for all $s, t \in [0,1]$. The advantage of the latter definition is that it is completely described in terms of the distance function $d_g$ and in particular does not require explicit mention of the Riemannian structure of the space. \nc
}

We can now define the notion of $\alpha$-\textit{geodesic convexity}. {We say  $V: \M \rightarrow \R$} is $\alpha$-geodesically convex if for all {$\theta, \theta' \in \M$ there exists a constant speed geodesic
$\gamma: [0,1] \rightarrow \M$}
between them, such that
{
\begin{equation}
   V \bigl(\gamma(t)\bigr) \leq t V(\theta) + (1-t) V(\theta') - \frac{\alpha}{2} t(1-t) d_g(\theta, \theta')^2,   
   \label{eq:GeoConv}
\end{equation}
}
for all $t \in [0,1]$.

Notice that $\alpha$-geodesic convexity reduces to $\alpha$-strong convexity when $(\M, g)$ is an Euclidean space. Also, it can be shown that $\alpha$-geodesic convexity for $\alpha >0$ implies the generalized PL condition \eqref{eq:PLGeneral}  (see Lemma 11.28 in \cite{boumal2020introduction}), which in turn implies, following the proof of Proposition \ref{prop:ExpDecay}, exponential decay rates for the energy {$V$} along its gradient flow, in direct analogy with Proposition \ref{prop:ExpDecay}. 

\begin{remark}
\label{rem:InnerProds}
Equation \eqref{eq:GeoDistance} relates the distance function $d_g$ with the family of inner products $g$. This formula is very useful as it allows us to recover the metric $g$ from its distance function $d_g$. This observation will be relevant when discussing the formal Riemannian structures on spaces of probability measures in the context of sampling
in section \ref{sec:LangevinGradFlow}.
\end{remark}

\subsection{Time Discretizations}\label{sec:discODE}

\subsubsection{Standard Gradient Descent}
\label{sec:OptimTimeDiscrete}
In this section we discuss how to obtain practical optimization algorithms by discretizing the gradient system \eqref{eq:GradDescent} in time. First, the \textit{explicit Euler} scheme gives the standard gradient descent iteration
\begin{equation}\label{eq:ExplicitEuler}
   \theta_{n+1}= \theta_n - \tau \nabla V(\theta_n).
\end{equation}
In numerical analysis of differential equations, $\tau>0$ is interpreted as a small time-step; then, if  \eqref{eq:ExplicitEuler} and \eqref{eq:GradDescent} are initialized at the same point $\theta_0$, it holds that $\theta_n \approx \theta_t$ for $t = n \tau.$
In the optimization context of interest, $\tau$ is referred to as a learning rate and it is insightful to note that \eqref{eq:ExplicitEuler} can be defined variationally as
\[ \theta_{n+1}= \argmin_{\theta} \Bigl( \langle \nabla V(\theta_n), \theta- \theta_n \rangle + \frac{1}{2\tau } | \theta- \theta_n |^2 \Bigr).   \]
Thus, $\theta_{n+1}$ is found by minimizing $V(\theta_n) + \langle \nabla V(\theta_n), \theta- \theta_n \rangle + \frac{1}{2\tau } | \theta- \theta_n |^2$, noticing that the first two terms form the first order approximation of the objective $V$ around the most recent iterate $\theta_n.$ In practice, the learning rate may be chosen adaptively using a line search \cite{nocedal}. 

Compared to the continuous-time setting, energy dissipation and convergence of the explicit Euler scheme require further assumptions on the function $V$. The following proposition is analogous to Proposition \ref{eq:ExpDecay} but relies on a \textit{smoothness} assumption on the gradient of $V$ additional to the PL condition. 
\begin{proposition}
\label{prop:DecayForDiscrete}
Suppose that $V$ has $L$-Lipschitz gradient, has minimum $V^*$, and satisfies the PL condition. Then the gradient descent algorithm defined by \eqref{eq:ExplicitEuler} with step-size $\tau := \frac{1}{L}$ has a linear convergence rate. More precisely, it holds that 
\begin{equation}
V(\theta_{n+1}) - V^* \le \Bigl(1 - \frac{\alpha}{L}\Bigr)^n \bigl(V(\theta_0) - V^*\bigr). 
\label{eq:DecayRatesTimeDiscrete}
\end{equation}
\end{proposition}
\begin{proof}
A classical result in convex analysis ensures that the assumption that $\nabla V$ is $L$-Lipschitz implies
$$V(\theta_{n+1}) \le V(\theta_n) + \langle \nabla V(\theta_n), \theta_{n+1} - \theta_n \rangle + \frac{L}{2} |\theta_{n+1} - \theta_n |^2.  $$
Using \eqref{eq:ExplicitEuler}, we then deduce that
$$V(\theta_{n+1}) \le V(\theta_n) - \frac{1}{2L} | \nabla V(\theta_n) |^2,$$ which combined with the PL condition gives
\begin{align*}
V(\theta_{n+1}) - V^* &\le V(\theta_n) - V^* - \frac{\alpha}{L} \bigl( V(\theta_n) - V^*)   \\
& = \Bigl(1 - \frac{\alpha}{L} \Bigr) \bigl( V(\theta_n) - V^* \bigr). 
\end{align*}
The result follows by induction. 
\end{proof}

 From the proof of Proposition \ref{prop:DecayForDiscrete} we see that, under the $L$-smoothness condition and assuming that the step size $\tau$ is sufficiently small, the explicit Euler scheme dissipates the energy $V$. Moreover, this condition helps us quantify the amount of dissipation in one iteration of the scheme in terms of the norm of the gradient of $V$ at the current iterate. 

As an alternative discretization, one can consider the \textit{implicit Euler} scheme:
\begin{equation}
 \theta_{n+1}= \theta_n - \tau \nabla V(\theta_{n+1}),
   \label{eq:ImplcitEuler2}
\end{equation}
which coincides with the first order optimality conditions for 
\begin{equation}
   \theta_{n+1}= \argmin_{\theta} \Bigl( V(\theta) + \frac{1}{2\tau } | \theta- \theta_n |^2 \Bigr). 
   \label{eq:ImplcitEuler1}
\end{equation}
It follows directly from the definition of the implicit Euler scheme that it satisfies a dissipation inequality analogous to \eqref{eq:EDI} without imposing any additional smoothness conditions on $V$. However,  an important caveat is that \red determining $\theta_{n+1}$ from $\theta_n$ requires \nc solving  a new optimization problem
\eqref{eq:ImplcitEuler1} 
or finding a root for the (in general) non-linear equation \eqref{eq:ImplcitEuler2}. On the other hand, from a theoretical perspective the implicit Euler scheme, or \textit{minimizing movement scheme} as it is called in \cite{ambrosio2008gradient}, is an important tool for proving existence of gradient flow equations in general metric spaces; see Chapters 1-2 in \cite{ambrosio2008gradient}.

\subsubsection{Discretizations and Preconditioning}
One possible time discretization for \eqref{eq:PreCondGradDescent} is given by
\[ \theta_{n+1}= \theta_n - \tau  H(\theta_n)^{-1}\nabla V(\theta_n),  \]
which is a direct adaptation of \eqref{eq:ExplicitEuler} to the preconditioned setting. \red  Proposition \ref{prop:DecayForDiscrete} \nc can be readily adapted using the PL condition and $L$-smoothness of $V$ relative to the geometry induced by the field $H$.

\begin{remark}
\label{rem:FasterRatesVsDiscretization}
In line with the discussion at the end of section \ref{sec:Precondition}, we notice that the effect of scaling the field $H$ by a constant $1/r$ is to scale the {constants in both the PL condition and the $L$-smoothness condition by a factor of $r$. The net gain in \eqref{eq:DecayRatesTimeDiscrete} from rescaling the metric is thus null.  }  
\end{remark}

Another \red possible time discretization for \eqref{eq:PreCondGradDescent} \nc when the function $H$ is the Hessian of a strictly convex function $h: \Theta \rightarrow \R$ (not necessarily equal to the objective function $V$) is the \textit{mirror descent} scheme:
\begin{equation}
    \begin{cases} z_{n+1} = z_n - \tau \nabla V(\theta_n),\\
\theta_{n+1}= (\nabla h)^{-1}(z_{n+1}).
\end{cases}
\label{eq:MirrorDescent}
\end{equation}
 The idea in mirror descent is to update an associated mirror variable (a transformation of $\theta$ by a mirror map, in this case $\nabla h $) using a gradient step, as opposed to directly updating the variable $\theta$ as in the standard explicit Euler scheme. Using a Taylor approximation of $(\nabla h)^{-1}$ around $z_t$ we see that
\begin{align*} 
\theta_{n+1} &= (\nabla h)^{-1} \bigl(z_n - \tau \nabla V (\theta_n) \bigr) 
\\ &\approx \theta_n - \tau H(\theta_n)^{-1} \nabla V (\theta_n), 
\end{align*}
revealing why mirror descent can be regarded as an approximation of \eqref{eq:PreCondGradDescent} when $H$ is the Hessian of $h$. 

It is worth remarking that the update rule \eqref{eq:MirrorDescent} has the following variational characterization 
\begin{equation}
  \theta_{n+1} := \argmin_{\theta \in \Theta} \Bigl( \langle \nabla V(\theta_n) , \theta  \rangle +\red  \frac{1}{\tau} \nc D_h(\theta \| \theta_n)\Bigr),    
  \label{eq:VariationalMirror}
\end{equation}
where the function $D_h(\theta \|\theta')$ has the form
\[ D_h(\theta \| \theta') := h(\theta) - h(\theta') - \langle \nabla h (\theta') , \theta -\theta' \rangle,   \]
and is often referred to as \textit{Bregman divergence}. This variational characterization was discovered and used in \cite{BECK2003167} to deduce convergence properties of mirror descent. Notice that the strict convexity of $h$ guarantees that the function $D_h(\cdot \|\cdot)$ is non-negative and zero only when both of its arguments coincide. Bregman divergences thus play a similar role to the one played by the quadratic function $\frac{1}{2} | \cdot - \cdot |^2$ in the variational form of the standard explicit Euler scheme.


\begin{figure}
    \centering
    \includegraphics[width=0.9\linewidth]{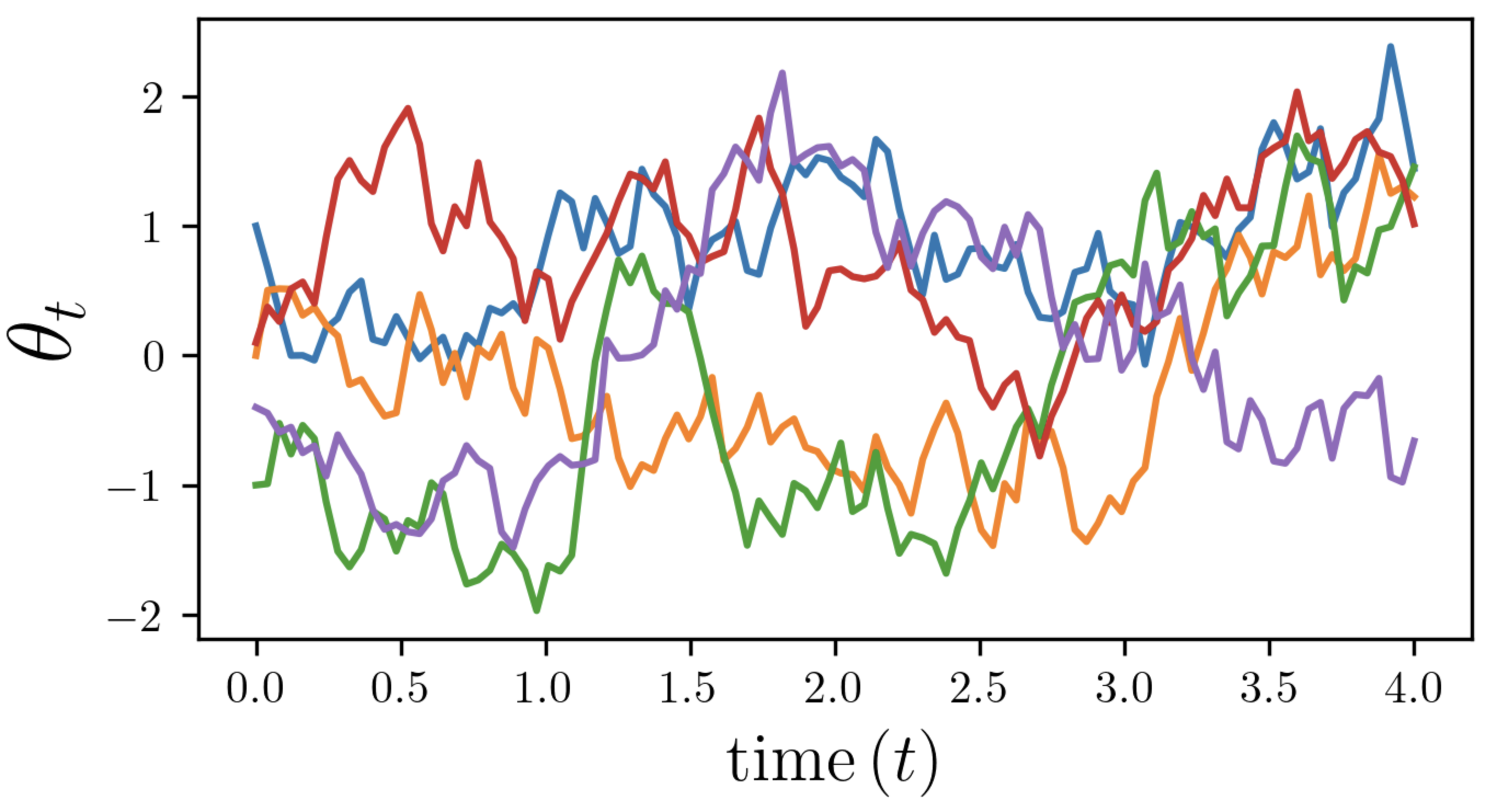}
   \vspace{-2ex} \caption{Five trajectories of Langevin dynamics with double well potential $V$ given by \eqref{eq:doublewell}.}
    \label{fig:langevintrajectories}
\end{figure}

\section{Sampling}
\label{sec:Sampling}
While direct sampling from certain distributions, e.g. Gaussians, may be rather straightforward, sampling from a general target distribution can be challenging, especially in high-dimensional settings. In this section we consider the problem of sampling a target density $\pi(\theta) \propto \exp\bigl( - V(\theta)\bigr).$ As guiding examples, one may consider sampling a posterior or Boltzmann distribution, see \eqref{V-definition}--\eqref{eq:boltzmann}. We start by introducing Langevin dynamics in section \ref{ssec:Langevin}, a stochastic differential equation that resembles the gradient system \eqref{eq:GradDescent}, but which incorporates a Brownian motion that makes the solution trajectories $\{\theta_t\}_{t \ge 0}$ random. In section \ref{ssec:convergenceLangevin} we present some results that state that, under suitable assumptions, the density $\rho_t$ of $\theta_t$ converges to the desired target density $\pi$ as $t\to \infty.$ In section \ref{ssec:langevinflows} we discuss how the Langevin dynamics define a gradient flow in the space of probability distributions. Finally, section \ref{ssec:discretizationLangevin} discusses how to use discretizations of Langevin dynamics to obtain practical sampling algorithms. Our presentation here parallels that of section \ref{sec:GradDescent}.


\subsection{Langevin Dynamics}\label{ssec:Langevin}
Consider the \emph{overdamped Langevin} diffusion \cite{pavliotis2014stochastic} 
\begin{equation}
   d\theta_t = - \nabla V(\theta_t) \, dt + \sqrt{2} \, dB_t,
   \label{eq:Langevin}
\end{equation}
where $\{ B_t \}_{t \geq 0}$ is a Brownian motion on $\Theta =\R^d$. Langevin dynamics can be interpreted as a stochastic version of the gradient descent dynamics \eqref{eq:GradDescent}. This is illustrated in the following example, which also provides intuition on the connection between Langevin dynamics and sampling.

\begin{figure}
    \centering
    \includegraphics[width =0.9\linewidth]{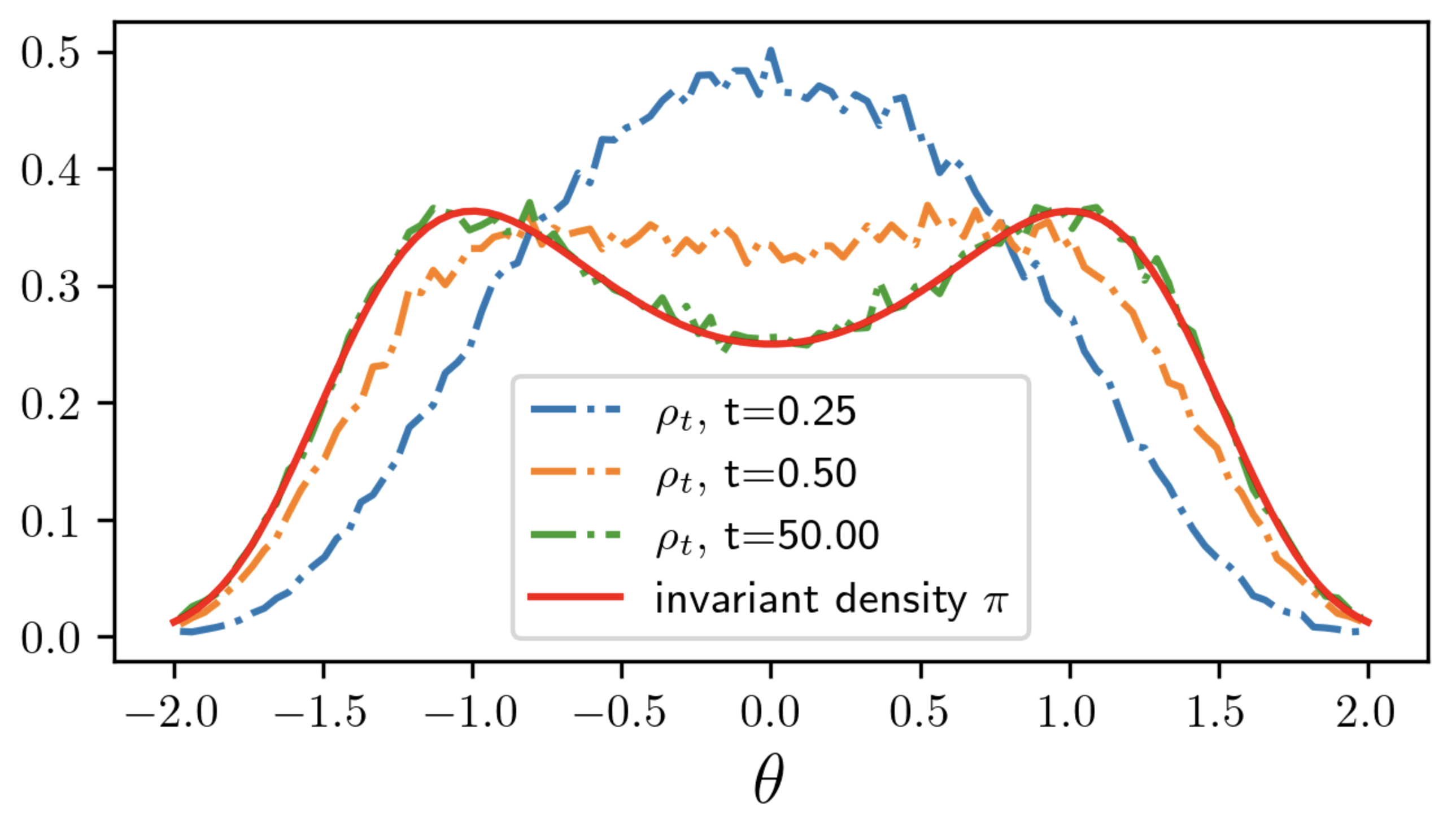}
     \vspace{-3ex} 
    \caption{Histograms of $\rho_t$ at $t = 0.25, 0.5, 50.$ For large $t,$ $\rho_t$ is close to the target density $\pi \propto e^{-V}.$}
    \label{fig:histograms}
\end{figure}

\begin{example}\label{ex:langevin}
 Consider Langevin dynamics with the double well potential $V$ introduced in \eqref{eq:doublewell}. Figure \ref{fig:langevintrajectories} shows five trajectories, initialized as in Figure \ref{fig:deterministic}. For each $t> 0,$ $\theta_t$ is now a random variable, whose Lebesgue density will be denoted by $\rho_t$ in what follows. Figure \ref{fig:histograms} shows an approximation of $\rho_t$ for $t \in\{ 0.25, 0.5,50\}$ obtained by simulating $N = 10^5$ solution trajectories. Notice that at $t= 50,$ $\rho_t$ is exceedingly close to the target density $\pi(\theta) \propto \exp\bigl( -V(\theta) \bigr),$ so that $\theta_t$ can be viewed as a sample from $\pi.$
 Thus, while $\theta_t$ is random due to  the Brownian motion, the density $\rho_t(\theta)$ is larger at points $\theta$ where  $V(\theta)$ is small. \end{example}

\subsection{A Note on Convergence}\label{ssec:convergenceLangevin}
It is natural to ask whether the law $\{\rho_t \}_{t\geq 0}$ of a given stochastic process
 converges to an \emph{invariant} distribution $\pi.$ For the Langevin diffusion, the positive answer illustrated in Example \ref{ex:langevin} holds under suitable assumptions on $V$ that are analogous to the PL and strong convexity conditions in section \ref{sec:GradDescent}.
 A natural way to study the long-time behavior of $\rho_t$ is to derive a differential equation for its evolution. To that end, one may characterize the time derivative of the action of $\rho_t$ on suitable test functions $\phi : \Theta \rightarrow \R$. More precisely, we compute $\frac{d}{dt} \int \phi(
\theta) d\rho_t (\theta)$, which is a standard derivative of a function from the real line to itself. For the Langevin diffusion \eqref{eq:Langevin} it can be proved that:
\begin{equation}
    \frac{d}{dt} \int_{\Theta} \phi(\theta) d \rho_t(\theta) = \red -\nc \int_{\Theta} \nabla \phi \cdot \nabla \bigl(V + \log(\rho_t)\bigr) d\rho_t(\theta),
    \label{eq:FokkerPlanckLangevinWeak}
\end{equation}
$ \forall t >0, \forall \phi \in C_c^\infty(\Theta).$ The above condition is the \textit{weak formulation} of the Fokker-Planck equation:
\begin{align}
   \partial_t \rho_t &=  \text{div} \bigl( \rho_t \nabla \bigl( V +  \log(\rho_t) \bigr) \bigr) =: \mathcal{L} \rho_t.
   \label{eq:FokkerPlanckLangevin}
\end{align}
From now on we  interpret \eqref{eq:FokkerPlanckLangevin} in its weak form \eqref{eq:FokkerPlanckLangevinWeak}. 

We observe that \red$\pi \propto e^{- V} $ \nc is a stationary point of the dynamics \eqref{eq:FokkerPlanckLangevin}. That is, if we initialize the dynamics at $\rho_0= \pi$, then $\rho_t:= \rho_0$ for all $t>0$ is a solution to the equation. \red The next result describes the long-time behavior of a solution to the Fokker-Planck equation when initialized at more general $\rho_0$.\nc 

\begin{theorem}
\label{thm:DissipationPoincare}
Let $\rho_t$ be the solution to the Fokker-Planck equation with $\rho_0 \in L^2(\pi^{-1}), $ where $L^2(\pi^{-1})$ is the $L^2$ space with the weight function $\pi^{-1}$. Suppose that there is $\alpha >0 $ such that the following \red Poincar\'e inequality holds: \nc for every $f \in C^1 \cap L^2(\pi)$ that has zero mean under $\pi,$ it holds that $\alpha \| f\|_{L^2(\pi)}^2 \le \| \nabla f \|_{L^2(\pi)}^2.$ Then it holds that
$$ \| \rho_t - \pi \|_{L^2(\pi^{-1})} \le e^{-\alpha t} \| \rho_0 - \pi\|_{L^2(\pi^{-1})}.$$
\end{theorem}
\begin{proof}
Define $u_t$ by $\rho_t = u_t \pi.$  We can verify that 
$$ \partial_t u_t = - \nabla V \cdot \nabla u_t + \dive( \nabla u_t), \quad u_0 = \rho_0 \pi^{-1}.$$
Therefore, the zero-mean function  $u_t -1$ satisfies 
\begin{equation}
    \frac{\partial (u_t - 1)}{\partial t} = \mathcal{L} (u_t-1).
\end{equation}
Multiplying by $(u_t-1) \pi,$ integrating, and using that by assumption
$\alpha \| u_t -1 \|_{L^2(\pi)}^2 \le \| \nabla u_t \|_{L^2(\pi)}^2,$ we deduce that 
\begin{equation}
    \frac{1}{2} \frac{d}{dt} \| u_t -1 \|_{L^2(\pi)}^2 \le - \alpha \| u_t - 1 \|_{L^2(\pi)}^2. 
\end{equation}
  Gronwall's inequality gives the desired result.
\end{proof}

The above result implies that, as time $t$ goes to infinity, the distribution $\rho_t$ converges toward the target density $\pi\propto e^{-V}$ \red exponentially fast. \nc The notion of convergence implied by Theorem \ref{thm:DissipationPoincare}, however, is not as strong as other notions such as the convergence in Kullback-Leibler (KL) divergence that will be discussed in the next section. In particular, Theorem \ref{thm:DissipationPoincare} should be contrasted with the discussion in section \ref{sec:GeoConvWass}.

\subsection{Choice of Objective and Metric}\label{ssec:langevinflows} 
Here we discuss a concrete variational interpretation for the Langevin system \eqref{eq:FokkerPlanckLangevin}. In essence, this entails viewing sampling as an optimization algorithm (in particular, as a gradient flow) that aims at recovering the target density $\pi$. As discussed toward the end of section \ref{sec:Precondition}, to realize this interpretation it is important to identify precisely the geometric objects involved in the definition of a gradient flow (energy, metric, etc.). In all subsequent sections, $\M$ will be the space of probability measures over $\Theta$, which we will denote with $ \mathcal{P}(\Theta)$, \red and $F$ will be the \textit{KL divergence} relative to $\pi$, which we recall is defined as: 
\begin{equation}\label{eq:MinGradFlow}
 \dkl(\nu \| \pi) :=  \int_{\Theta}  \log\Bigl( \frac{\nu(\theta)}{ \pi(\theta)}\Bigr) \nu(\theta)\, d\theta, \,   
 \end{equation} 
 whenever $\nu$ is absolutely continuous w.r.t. $\pi$ and $\dkl(\nu \| \pi) = \infty$ otherwise.
\nc Different choices of metric over $\mathcal{P}(\Theta)$ will induce different evolution equations (recall our discussion of preconditioning in the context of optimization over the parameter space $\Theta$), and thus different 
optimization schemes. In section \ref{sec:LangevinGradFlow} we discuss a specific geometric structure for $\mathcal{P}(\Theta)$ that realizes the Fokker-Planck equation of the Langevin diffusion as a gradient flow of $F$, and in section \ref{sec:otherGFs} we discuss other gradient flow structures that motivate other sampling algorithms.

\begin{remark}
\label{rem:VariationalInference}
It holds that $\dkl(\nu\| \pi)\geq 0$ for all $\nu$, and equality holds if and only if $\nu=\pi$. In particular, the unique minimizer of $ \nu\in  \mathcal{P}(\Theta) \mapsto \dkl(\nu \| \pi)$ is the target $\pi$. Although trivial, this observation motivates  \textit{variational inference} --\red see 
\cite{lambert2022variational} and references therein\nc--
a method for producing tractable proxies for $\pi$ that relies on the minimization of $\dkl(\cdot||\nu)$ over
 a user-chosen family of tractable distributions. 
\end{remark}

\begin{remark}
If $\pi$ has a density w.r.t. the Lebesgue measure that is proportional to $e^{-V}$, then $\dkl(\nu \| \pi)<\infty$ implies that $\nu$ is also absolutely continuous w.r.t. the Lebesgue measure. In that case we will abuse notation slightly and use $\nu$ to also denote $\nu$'s corresponding density. In particular, when we write $\log(\nu)$ it is understood that $\nu$ is interpreted as the density function w.r.t. Lebesgue measure of the measure $\nu$.
\end{remark}


\subsubsection{KL and Wasserstein}
\label{sec:LangevinGradFlow}

Following a series of seminal works that started with a paper by Jordan, Kinderlehrer, and Otto in the late 90's (see sections 8.1. and 8.2. in \cite{villani2003topics}) we will interpret equation \eqref{eq:FokkerPlanckLangevin} as the gradient flow of the energy $F$ w.r.t the \textit{Wasserstein} metric. 
Given $\rho, \rho' \in \mathcal{P}(\Theta)$ with finite second moments, their Wasserstein distance $W_2(\rho, \rho')$ is given by 
\begin{equation}
W^2_2(\rho, \rho'):= \min_{\Upsilon \in \Gamma(\rho, \rho') } \int_{\Theta \times \Theta} |\theta- \theta'|^2 d \Upsilon(\theta,\theta'),
\label{eq:Wasserstein}
\end{equation}
where $\Gamma(\rho, \rho')$ is the set of couplings between $\rho$ and $\rho'$, i.e. the set of Borel probability measures on the product space $\Theta \times \Theta$ with first and second marginals equal to $\rho$ and $\rho'$, respectively.

Formula \eqref{eq:Wasserstein}, although simple, does not reveal the infinitesimal geometric structure of the Wasserstein space to define gradients of functionals over $\mathcal{P}(\Theta)$. What is missing is a representation of the distance $W$ in a form similar to \eqref{eq:GeoDistance}. The next result by Benamou and Brenier (see section 8.1. in \cite{villani2003topics}) provides the missing elements.


\begin{proposition}
Let $\rho, \rho' \in \mathcal{P}(\Theta)$. Then
\begin{align}\label{eq:DynamicOT}
W^2_2(\rho, \rho') =  &  \inf_{t \in [0,1] \mapsto (\gamma_t,  \nabla \varphi_t)} \int_0^1 \int _\Theta | \nabla \varphi_t (\theta)|^2  d\gamma_{t}(\theta)   dt \notag
\\& {\emph{s.t.}}\:  \partial_t \gamma_t + \emph{div} (\gamma_t \nabla \varphi_t ) = 0, 
\\& \gamma(0) = \rho, \gamma(1) = \rho'. \notag
\end{align}
\label{prop:BenamouBrenier}
\end{proposition}
The infimum in \eqref{eq:DynamicOT} is taken over all maps $t \in [0,1] \rightarrow (\gamma_t, \nabla \varphi_t)$, where each pair $(\gamma_t, \nabla \varphi_t)$ consists of a measure $\gamma_t \in \mathcal{P}(\Theta)$ and a vector field of the form $\nabla \varphi_t$ for a smooth $\varphi_t:\Theta \rightarrow \R$, that satisfy the \textit{continuity equation}:
\begin{equation}
\partial_t \gamma_t + \textrm{div} (\gamma_t \nabla \varphi_t ) = 0,  
\label{eq:ContEqu}
\end{equation}
interpreted in weak form.
Notice that the Fokker-Planck equation \eqref{eq:FokkerPlanckLangevin} is a particular case of the continuity equation with $\varphi_t= V + \log(\rho_t)$. Inspired by equation \eqref{eq:GeoDistance} (see also Remark \ref{rem:InnerProds}) we can provide a geometric interpretation of identity \eqref{eq:DynamicOT}: the continuity equation \eqref{eq:ContEqu} provides a representation of curves in the formal manifold $\M = \mathcal{P}(\Theta)$. In this representation, the velocity of a curve (tangent vector) at each point in the curve can be identified with a vector field (over $\Theta$) of the form $\nabla \varphi$. Furthermore, \eqref{eq:DynamicOT} motivates introducing an inner product at each $\nu \in \mathcal{P}(\Theta)$ of the form: 
\[ g_\nu (\nabla \varphi , \nabla \varphi'):= \int_{\Theta} \nabla \varphi(\theta) \cdot \nabla \varphi' (\theta) \, d\nu(\theta). \]


With the above geometric interpretation in place, we may follow equation \eqref{eq:DefGradient} and identify the gradient of $F(\cdot)= \dkl(\cdot||\pi)$ at an arbitrary point $\nu$. For this purpose take an arbitrary solution $(\gamma_t, \nabla \varphi_t)$ to the continuity equation \eqref{eq:ContEqu} (i.e. take an arbitrary curve in $\M$) for which $F(\gamma_t)<\infty$ and compute: 
 \begin{align*}
   \frac{d}{dt} F(\gamma_t) &= \frac{d}{dt} \int_{\Theta} \log \left( \frac{\gamma_t}{e^{-V}} \right) d\gamma_t(\theta)
   \\& = \int_{\Theta} \nabla \varphi_t \cdot \nabla \bigl( V + \log(\gamma_t) \bigr) \,  d \gamma_t(\theta)
   \\& = g_{\gamma_t} \bigl( \nabla \varphi_t , \nabla ( V + \log(\gamma_t) ) \bigr).
 \end{align*}
 In the above we have gone from the first line to the second one using the weak form of the continuity equation; \red to go from the second to third line we have used the definition of $g_{\gamma_t}$. \nc

From the above computation we conclude that the gradient of $F$ (w.r.t. to the Wasserstein metric) at a point $\nu$ for which $F(\nu)<\infty$ takes the form $\nabla \bigl( V + \log(\nu) \bigr)$. In particular, the curve in $\M = \mathcal{P}(\Theta)$ whose velocity vector agrees with the negative gradient of the functional $F$ takes the form of the Fokker-Planck equation \eqref{eq:FokkerPlanckLangevin}. In other words,  \eqref{eq:FokkerPlanckLangevin} can be interpreted as the gradient flow of $F(\cdot)= \dkl(\cdot\| \pi ) $ w.r.t. the Wasserstein metric.

\begin{remark}
To some extent, the computations in this section have been formal and some of the above derivations have been left unjustified. These computations rely on a formal adaptation of formula \eqref{eq:GradDescent} to the setting of the Riemannain manifold $\mathcal{P}(\Theta)$ endowed with the Wasserstein distance. For a rigorous treatment of the topics discussed in this section the reader is referred to the second part of the book \cite{ambrosio2008gradient}.  There, the notion of gradient flow in $\mathcal{P}(\Theta)$ is motivated by the dissipation identity \eqref{eq:EDE} and adapted to the metric space $(\mathcal{P}(\Theta), W_2)$.
\end{remark}

\begin{remark}
At a high level, the ideas discussed in this section can be used to propose flows aimed at solving variational inference problems like the ones briefly mentioned in Remark \ref{rem:VariationalInference}. Indeed, following the geometric intuition from projected gradient descent methods, where one uses the projection of the negative gradient of the objective onto the tangent plane of the constrained set to define the projected gradient descent flow, one may consider the projection of the negative gradient of the energy $F$ (w.r.t Wasserstein) onto the tangent planes of the submanifold $\mathcal{G}$; naturally, in this setting the notion of (orthogonal) projection is taken w.r.t. the Riemannain metric underlying the Wasserstein space. This idea has been recently explored in \cite{lambert2022variational} for certain families $\mathcal{G}$ of tractable distributions. 
\end{remark}

\subsubsection{Geodesic Convexity of the Relative Entropy in the Wasserstein Space}
 \label{sec:GeoConvWass}
 
The definition of geodesic convexity introduced in \eqref{eq:GeoConv} can be readily adapted to the setting of an energy defined over an arbitrary metric space, and in particular to the setting of $\mathcal{P}(\theta)$ endowed with the Wasserstein distance. Indeed, notice that equation \eqref{eq:GeoConv} is completely determined by the energy of interest, the distance function, and the notion of constant speed geodesic, which in turn can be defined in terms of the distance function. 

We have the following theorem by McCann relating the convexity of the function $V$ with the $\alpha$-geodesic convexity of $\dkl(\cdot \| \pi)$ when $\pi \propto e^{-V}$; see Theorem 5.15. in \cite{villani2003topics}.
 \begin{theorem}
 Suppose that $V$ is $\alpha$-strongly convex and let $\pi \propto e^{-V}$. Then  $\dkl(\cdot \| \pi)$ is $\alpha$-geodesically convex w.r.t. the Wasserstein distance.
\label{thm:GeodesicConvexityRelEntropy}
 \end{theorem}

With Theorem \ref{thm:GeodesicConvexityRelEntropy} in hand, the $\alpha$-strong convexity of $V$ implies that
\[ \dkl(\rho_t\| \pi) \leq e^{-2 \alpha t} \dkl(\rho_0\| \pi) , \quad \forall t \geq 0,  \]
along a solution $\{ \rho_t \}_{t \geq 0}$ of the Fokker-Planck equation \eqref{eq:FokkerPlanckLangevin}. \red This notion of exponential contraction to equilibrium is not implied by the Poincar\'e condition from Theorem \ref{thm:DissipationPoincare}. \nc

\nc


\subsection{Time Discretizations}\label{ssec:discretizationLangevin}
This section describes how to obtain sampling algorithms from time discretization of the Langevin dynamics \eqref{eq:Langevin}.
In analogy to the explicit Euler scheme \eqref{eq:ExplicitEuler} for \eqref{eq:GradDescent}, the Euler-Maruyama discretization for \eqref{eq:Langevin} is given by $\theta_0 \sim \rho_0,$ and
\begin{equation}\label{eq:EM}
   \theta_{n+1}= \theta_n - \tau \nabla V(\theta_n) + \sqrt{2 \tau} \xi_n, \quad \xi_n \overset{\text{i.i.d.}}{\sim} N(0,1).
\end{equation}
For $t = n\tau,$ the law $\rho_n$ of $\theta_n$ approximates the law $\rho_t$ of $\theta_t$ given by \eqref{eq:Langevin}. However, the error introduced by time discretization causes $\rho_n$ to not converge, in general, to $\pi$ as $n \to \infty.$ In other words, the probability kernel 
\bh{$q(\theta_n, \cdot) = \text{law}(\theta_{n+1}| \theta_n)$}, defined by the Markov chain \eqref{eq:EM}, does not leave $\pi$ invariant. 

To remedy this issue, one may consider using \eqref{eq:EM} as a \emph{proposal kernel} within a Metropolis-Hastings algorithm \cite{liu2001monte}, leading to the \emph{Metropolis Adjusted Langeving Algorithm}, often referred to as MALA. The basic idea is to use an accept/reject mechanism to turn the proposal kernel $q$ into a new Markov kernel that leaves $\pi$ invariant. 
Given the current state $\theta_n,$ one proposes a move $\theta_n \mapsto \theta_{n+1}^*$ by sampling $q(\theta_n, \cdot);$ the move is accepted with a probability
\begin{equation}\label{eq:acceptance}
    a = \min \biggl(1, \frac{\pi(\theta_{n+1}^*)}{\pi(\theta_n)} \frac{q(\theta_{n+1}^*,\theta_n)}{q(\theta_n,\theta_{n+1}^*)} \biggr).
\end{equation}
If the move is accepted, we set $\theta_{n+1} :=\theta_{n+1}^*.$ Otherwise, we set $\theta_{n+1} := \theta_n.$ The Metropolis-Hastings acceptance probability \eqref{eq:acceptance} is chosen in such a way that $\pi$ is the invariant distribution of the new chain $\{\theta_n\}_{n=1}^\infty.$ Notice that the two steps of the algorithm, namely, sampling from the proposal kernel defined by \eqref{eq:EM} and evaluating \eqref{eq:acceptance}, can be implemented without knowledge of the normalizing constant of $\pi.$ 
We refer to 
\cites{roberts1996exponential}
for further details on the convergence of Langevin diffusions and their discretizations and we refer to \cite{chewi2020exponential} for a sampling analog of the mirror descent optimization algorithm in \eqref{eq:MirrorDescent}.

\section{Modern Twists on Langevin}\label{sec:otherGFs}
In this section we outline \red recent \nc extensions of the gradient flows of section~\ref{sec:LangevinGradFlow} aimed at sampling.  The idea is to employ gradient flows 
of the energy $F = \dkl(\cdot \| \pi)$  w.r.t. metrics beyond the Wasserstein distance. One hopes that the new dynamics lead to faster convergence to minimizers and, in turn, to more efficient sampling algorithms.

\subsection{Ensemble Preconditioning}
In analogy with section~\ref{sec:Precondition}, we consider 
preconditioned variants of the Langevin diffusion \eqref{eq:Langevin},  
\begin{equation}\label{eq:preconditioned-Langevin1}
    d \theta_t = -  H(\theta_t)^{-1} \nabla V(\theta_t) \, dt + \sqrt{2 H(\theta_t)^{-1} } 
    \, dB_t,
\end{equation}
where $H(\theta)$ is once again the preconditioning matrix field. The intuition from 
subsection~\ref{sec:Precondition} carries over in this setting: by choosing an 
appropriate pre-conditioner we can speedup the convergence of the Langevin diffusion 
to the target density. However, in contrast to \eqref{eq:Langevin}, choosing $H$ as a function of $\theta_t$ can in general lead to a nonlinear evolution for the law of the process. \red Moreover, without additional structure, the target $\pi\propto e^{-V}$ may not be an invariant distribution for \eqref{eq:preconditioned-Langevin1}. This motivates discussing suitable choices for $H$.   \nc 

One approach to constructing the matrix $H$ is to consider an ensemble of particles 
evolving \red according to \eqref{eq:preconditioned-Langevin1} and use \nc the location of 
the particles to construct an appropriate preconditioner. 
Following  \cite{garbuno2020interacting}, consider an ensemble 
of $J \ge 1$ interacting Langevin
diffusions $ \theta_t := \{ \theta_t^{(j)} \}_{j=1}^J$ that are evolved according to the 
coupled system  
\begin{equation*}
    d \theta^{(j)}_t = - H( \theta_t)^{-1} \nabla V(\theta_t^{(j)}) + \sqrt{2 H(\theta_t)^{-1}} \, d B_t^{(j)},
\end{equation*}
where $\{ B_t^{(j)} \}_{j=1}^J$ are i.i.d. Brownian motions, and 
\begin{equation*}
    H(\theta_t) = \frac{1}{J} \sum_{j=1}^J \left( \theta_t^{(j)} - \bar{\theta}_t \right)
     \left( \theta_t^{(j)} - \bar{\theta}_t \right)^\top,
\end{equation*}
for $\bar{\theta}_t$ the ensemble mean of the $\theta_t^{(j)}$. In other words, the 
precoditioner is chosen to be the empirical covariance matrix of the  ensemble at each point 
in time, a quantity that is convenient to compute in practice. 


Similarly to previous sections, one can show that the
ensemble preconditioned Langevin dynamics has a gradient flow structure. Taking the mean-field limit, i.e., letting $J \to \infty$, we formally obtain 
the following diffusion for the evolution of the ensemble 
\begin{equation*}
    d \theta_t = - H(\rho_t)^{-1} \nabla V(\theta_t) \, dt + \sqrt{2 H(\rho_t)^{-1}} \, dB_t, 
\end{equation*}
where $\rho_t$ denotes the distribution of $\theta_t$ as before and $H(\rho_t)$
is the covariance matrix of $\rho_t$.  The Fokker-Planck equation for $\rho_t$ is given by
\begin{equation}
    \partial_t \rho_t = {\rm div} \Bigl( \rho_t H(\rho_t)^{-1} \big( \nabla V(\theta_t) + \nabla \log \rho_t \big) \Bigr),
    \label{eqn:WeightedFP}
\end{equation}
\red which has the target $\pi\propto e^{-V}$ as a stationary point. \nc The divergence form in \eqref{eqn:WeightedFP} suggests a gradient flow structure as shown 
in \cite{garbuno2020interacting}. Indeed, the above equation is a gradient flow of $F$ w.r.t. the \emph{Kalman-Wasserstein} distance $W_K$ defined as 
\begin{equation*}
\begin{aligned}
       W_{K}^2( \rho, \rho') & := \inf_{\gamma_t, \varphi_t} \int_0^1 \int \langle \nabla \varphi_t, 
    H(\rho_t) \nabla \varphi_t \rangle \rho_t \:  d \theta \: dt, \\
    & \text{s.t.} \:  \partial_t \gamma_t + {\rm div} ( \gamma_t H(\gamma_t) \nabla \varphi_t) = 0, \\
    &  \gamma(0) = \rho, \quad \gamma(1) = \rho'. 
\end{aligned}
\end{equation*}

\subsection{Langevin with Birth-Death}
\label{sec:LangevinBirthDeath}
One of the shortcomings of the Langevin dynamics is that its convergence 
suffers when sampling from multi-modal target densities: the process may get stuck around one of the modes and it may take a long time to cross the energy barrier between various modes or to overcome entropic bottlenecks. To ameliorate this metastable behavior
\cite{lu2019accelerating} proposes to consider the 
{\it birth-death accelerated Langevin (BDL)} dynamics
\begin{align}\label{eq:BDL}
\begin{split}
    \partial_t \rho_t &=  \mathcal{L} \rho_t  + \rho_t( \log \pi - \log \rho_t) \\&- \rho_t \int_\theta (\log \pi - \log \rho_t) d \rho_t(\theta),
\end{split}    
\end{align}
where $\mathcal{L}$ is as in \eqref{eq:Langevin}. Notice that compared with the Fokker-Planck equation \eqref{eq:FokkerPlanckLangevin}, equation \eqref{eq:BDL} contains two additional terms. The second term on the right-hand side 
favors increasing (resp. decreasing) $\rho_t(\theta)$ whenever $\rho_t(\theta) < \pi(\theta)$
(resp. $\rho_t(\theta) > \pi(\theta)$), hence the name birth-death. The third term is only included to ensure that $\rho_t$ remains a probability distribution through the birth-death process. The birth-death terms make the equation non-local (due to averaging) and allow the dynamics to explore the 
support of a multi-modal $\pi$ more efficiently: if the dynamics get stuck in one mode, 
the birth-death process can still transfer some mass to another mode that has not yet been thoroughly explored. 

It turns out that, akin to the Langevin dynamics, BDL is also a gradient flow of the KL
divergence but w.r.t. a modification of the Wasserstein distance
referred to as the {\it Wasserstein-Fisher-Rao (or the Spherical Hellinger-Kantorovich)} 
distance:
\begin{subequations}\label{eq:WFR}
\begin{align*} 
& W_{\mbox {\tiny{\rm FR}}}^2  ( \rho, \rho') \\
& \hspace{2ex} = \inf_{\gamma_t, \varphi_t} \int_0^1 \left( \int |\nabla \varphi_t |^2 +  |\varphi_t|^2  d \gamma_t - \Bigl( \int \varphi_t d \gamma_t  \Bigr)^2 \right) dt
\end{align*}
subject to the constraints 
\begin{align}
     & \partial_t \gamma_t + \textrm{div} (\gamma_t \nabla \varphi_t ) = -  \gamma_t  \Bigl( \varphi_t -  \int \varphi_t d \gamma_t   \Bigr), \label{eq:WFR-1} \\
     & \gamma(0) = \rho, \quad \gamma(1) = \rho'. \label{eq:WFR-2}
\end{align}
\end{subequations}
As before, the BDL continuity equation \eqref{eq:WFR-1} is understood in the weak sense and 
plays the same role as the continuity equation in the Benamou-Brenier formulation 
 \eqref{eq:DynamicOT}.
Equation \eqref{eq:WFR-1} thus provides an alternative representation for admissible curves in the space $\mathcal{P}(\Theta)$. Motivated by this, we will thus think  of a pair $(\varphi, \nabla \varphi)$ as tangent to a given point $\nu \in \mathcal{P}(\Theta)$ and introduce the inner product at the point $\nu$
 \begin{align*}
 g_\nu  \bigl( (\varphi, \nabla \varphi)& , (\varphi', \nabla \varphi') \bigr) :=  \int \nabla \varphi \cdot  \nabla \varphi ' d \nu 
 \\&+ \int \varphi \varphi' d \nu - \int \varphi  d \nu \int \varphi'  d \nu .
 \end{align*}
 

   With this geometric structure for the space $\mathcal{P}(\Theta)$ endowed with the WFR metric, we can proceed to carry out an analogous computation to the one at the end of section \ref{sec:LangevinGradFlow}  and show that the BDL dynamics is the gradient flow of the energy functional  $\dkl( \cdot \| \pi)$ w.r.t.  the $W_{\rm FR}$
  geometry; see \cite{lu2019accelerating}*{Thms.~3.2 and 3.3}. Furthermore, the rate of convergence of BDL is at least as good as that of Langevin dynamics
   and is asymptotically independent of the negative-log-density $V$, making it suitable for exploration of multi-modal landscapes.
 However, an important caveat is that turning BDL into a sampling algorithm requires utilizing an ensemble of interacting Langevin trajectories to empirically approximate the mean-field dynamics \eqref{eq:BDL} using a kernel density estimation, which may be forbiddingly expensive in high dimensional settings.  


\section{Conclusion and Discussion}\label{sec:conclusion}
This article provided a gentle introduction to gradient flows as a unifying framework for the design 
and analysis of optimization and sampling algorithms. 
Three key elements are involved in specifying a gradient flow: the space of interest (parameter or distribution); 
an energy function defined on that space to be minimized; and a geometric notion of a 
gradient. After discussing different versions of gradient flows for optimization, we mostly focused on Langevin dynamics 
as a sampling analog to gradient descent and discussed its generalizations 
through ensemble preconditioning and addition of non-local terms. 
 
The flexibility of the gradient flow framework gives significant freedom in 
the choice of the energy function and geometry  to design new algorithms 
beyond the examples discussed in this article. For instance,
in \cite{arbel2019maximum} the gradient flow of the maximum mean discrepancy (as opposed to KL) w.r.t. the Wasserstein distance  is used to analyze an ensemble 
sampling algorithm. As another example, the
\emph{Stein Variational Gradient Descent} or SVGD algorithm \cite{liu2017stein} can be viewed as a 
gradient flow of KL w.r.t. to a modified Wasserstein distance defined 
from an appropriate reproducing kernel Hilbert space.

To close, we point out that many alternative approaches that do not rely on direct reference to gradient flows can be used to enhance the convergence of Langevin dynamics and ameliorate their metastable behavior. For example, \emph{Hamiltonian Monte Carlo} or HMC (see \cite{liu2001monte} Chapter 10) utilizes a proposal kernel obtained by discretization of (deterministic) Hamilton equations as opposed to (stochastic) Langevin dynamics; theoretical and empirical evidence suggests that HMC scales favorably to high-dimensional settings. As another example, momentum methods may be used to accelerate the convergence of optimization and sampling algorithms; these methods may be interpreted as arising from time discretization of higher-order systems that approximate a gradient flow structure in certain limiting regimes \cite{kovachki2021continuous}.
Non-reversible variants of the Langevin diffusion can potentially achieve faster convergence (see section 4.8 in \cite{pavliotis2014stochastic}), while umbrella, tempering, and annealing sampling strategies (see chapter 10 in \cite{liu2001monte}) enable the efficient traversing of multi-modal targets. 
\bh{Finally, recent machine learning techniques such as normalizing flows 
\cite{kobyzev2020normalizing} offer an alternative approach to sampling 
and density estimation by direct parameterization of transport maps using 
neural networks. Employing the so called NeuralODE models leads to 
formulations that closely resemble gradient flows and the 
continuity equation \eqref{eq:ContEqu} with the vector field $\nabla \varphi_t$
replaced by a neural network.
}











	\paragraph{Acknowledgment} 
	NGT is supported by the NSF grant DMS-2005797.
	BH is supported by the NSF grant DMS-2208535.
	DSA is supported by the NSF grants DMS-2027056 and DMS-2237628, the DOE grant DOE DE-SC0022232, and the BBVA Foundation.

\bibliography{ExampleRefs}

\end{document}